\documentclass{llncs}

\usepackage{cite}
\usepackage{amsmath,amssymb,amsfonts}
\usepackage{algpseudocode, algorithm}
\usepackage{graphicx}
\usepackage{tabularx,multirow}

\usepackage{xspace}
\usepackage{float}
\usepackage{flushend}

\usepackage{textcomp}
\usepackage{xcolor}
\def\BibTeX{{\rm B\kern-.05em{\sc i\kern-.025em b}\kern-.08em
		T\kern-.1667em\lower.7ex\hbox{E}\kern-.125emX}}

\usepackage{geometry}
\usepackage{amsthm}
\theoremstyle{definition}

\DeclareMathOperator*{\argmax}{argmax}

\newcommand{\f}[1]{\texttt{\MakeUppercase{#1F}}}
\newcommand{\coins}[1]{\mathbf{c}(#1)}
\newcommand{\fee}[1]{\mathbf{f}(#1)}
\newcommand{\ebc}[1]{\mathbf{e}(#1)}
\newcommand{\bc}[1]{\mathbf{bc}(#1)}
\newcommand{\sender}[1]{\mathbf{s}(#1)}
\newcommand{\receiver}[1]{\mathbf{r}(#1)}

\title{How to profit from payments channels\thanks{This work with minor changes is accepted for publication in Financial Cryptography and Data Security (FC) 2020.}}

\author{O\u{g}uzhan Ersoy, Stefanie Roos and Zekeriya Erkin}
\authorrunning{O\u{g}uzhan Ersoy et al.}
\institute{Delft University of Technology\\
	\email{\{o.ersoy, s.roos, z.erkin\}@tudelft.nl} }

\begin{document}

\maketitle 

\begin{abstract}
Payment channel networks like Bitcoin's Lightning network are an
auspicious approach for realizing high transaction throughput and
almost-instant confirmations in blockchain networks. However, the
ability to successfully make payments in such networks relies on the
willingness of participants to lock collateral in the network. In
Lightning, the key financial incentive is to lock collateral are small
fees for routing payments for other participants. While users can choose these fees, currently, they mainly stick to the default fees.
By providing insights on
beneficial choices for fees, we aim to incentivize users to lock more
collateral and improve the effectiveness of the network.

In this paper, we consider a node $\mathbf{A}$ that given the network topology and the channel details selects where to establish channels and how much fee to charge such that its
financial gain is maximized. We formalize the optimization problem and show that it is NP-hard.
We design a greedy algorithm to approximate the optimal solution.
In each step, our greedy algorithm selects a node which maximizes the total reward concerning the number of shortest paths passing through $\mathbf{A}$ and channel fees. 
Our simulation study leverages real-world data set to quantify the impact of our gain optimization and indicates that our strategy is at least a factor two better than other strategies.   
\end{abstract}

\section{Introduction}\label{sec:intro}

%%motivate payment channels 
Payment channel networks~\cite{gudgeon2019sok} overcome the need to globally agree on every transaction in a blockchain. Instead, nodes can open and close \emph{channels} that they can use to transfer coins directly. In the absence of disputes, transactions only require local communication between the parties involved in a transaction. Nodes without a direct payment channel can route payments via intermediaries to avoid the transaction fees and delays for establishing a channel. 
By moving transactions off-chain, payment channels have the potential to drastically increase the transaction throughput while reducing the confirmation times from tens of minutes to sub-seconds. The most notable examples of payment channel networks are Bitcoin's Lightning~\cite{lightning} and Ethereum's Raiden~\cite{raiden}.

%and the need to lock collateral; 
When opening a payment channel, nodes need to lock coins that they cannot use outside of the channel during the lifetime of the channel. This opportunity cost makes it unattractive to maintain payment channels. However, routing payments in a network require that the network has well-funded channels~\cite{gudgeon2019sok}. The key incentives for locking collateral in a channel are i) frequent transaction with the other party~\cite{decker2015fast} and ii) financial gain through routing fees~\cite{DBLP:conf/middleware/EngelmannKKGW17}, i.e., fees that nodes charge for routing payments as intermediaries.  Our analysis on the Lightning network shows that the fees charged for routing are currently low and mainly equal to the default value~\cite{elementsproject}, we conjecture that the current payment channel networks primarily rely on the first incentive. However, research on the Lightning network suggests that this incentive entails networks of a low resilience with a few central hubs~\cite{rohrer2019discharged}. Analyzing the second incentives and show-casing that payments channels can entail financial profit is the most promising avenue of research to incentivize the participation in payment channel networks and fully leverage the potential of this promising blockchain scalability approach. 

%research question: 
In this paper, we adapt a payment channel network model based on Lightning. We assume a known topology and fees. Nodes select the cheapest path to conduct a payment. A node $\mathbf{A}$ aims to maximize its profit through routing fees by choosing both its payment channels and fees. The problem is challenging as higher fees indicate a higher profit if the node routes the payment but also a lower probability to be chosen for routing due to the transactions taking the cheapest path.  

%%some related work
Despite the importance of fees in payment channel networks, the issue has been mainly ignored in past research. The majority of papers deal with cryptographic protocols for channel establishment and multi-hop payments (e.g.,~\cite{spilmanbitcoinj,decker2015fast,egger2018amcu,decker2018eltoo,sprites}) as well as algorithms for routing payments (e.g., ~\cite{Flare2016,roos2017settling,hoenisch2018aodv}). 
There is some work on comparing routing fees to the on-chain fees of blockchains and presenting an economical analysis of the relation between the two fee types
~\cite{DBLP:conf/middleware/EngelmannKKGW17,DBLP:journals/corr/abs-1712-10222}. It is interesting to note that routing fees are related to the payment value whereas on-chain blockchain fees usually relate to the size of the transactions. 
In contrast, Di Stasi et al.~\cite{DBLP:conf/ithings/StasiACV18} evaluated the impact of routing fees on keeping channels balanced, i.e., ensuring that a channel is not used exclusively in one direction. The authors suggest a novel linear fee policy for each channel to improve channel balances. 
Most similar to our work, Avarikioti et al.~\cite{DBLP:conf/esorics/AvarikiotiJWW18} studied the optimal fee assignment of channels from the point of view of a payment service provider (PSP).
The authors analyzed optimal channel fees of the whole network that maximizes the total reward of the PSP instead of focusing on a node, which defines our problem.
In this manner, the authors presented a linear program solving the optimal fee assignment when the network is tree-structured, which is a very limiting assumption. 

%In \cite{DBLP:journals/corr/abs-1712-10222}, Br\^{a}nzei et al. presented an economical analysis of the channel fees in minors perspective. 
%More specifically, the authors investigated the relationship between Bitcoin transaction fees and the Lightning channel fees and computed Bitcoin fee equilibrium with and without Lightning network.
%In these analyses, the authors used special network topologies like star or uniform to model the Lightning network.

%In \cite{DBLP:conf/esorics/AvarikiotiJWW18}, Avarikioti et al. studied the optimal fee assignment of channels from the point of view of a payment service provider (PSP).
%PSP decides on the graph structure as well as the fees of channels in that graph to maximize its profit.
%As a result, the authors presented linear program solving optimal fee assignment when the graph is tree-structured.

%In \cite{subramanianrebalancing} Rebalancing in Acyclic Payment Networks*

%In \cite{DBLP:conf/ithings/StasiACV18}, Di Stasi et al. analysed the channel fee regarding the channel balances.
%The authors provide a linear fee policy for each channel, which targets to balance the channel.

%In \cite{DBLP:conf/middleware/EngelmannKKGW17}, Engelmann et al. modelled economic incentives of usage of the payment channels.
%They showed that for some cases using the Bitcoin network is cheaper than Lightning network.

%%our contribution 
We are hence the first to cover the aspect of maximizing fees in payment channel networks. More precisely, we formalize the problem of maximizing fees in a Lightning-inspired system model. We present an algorithm for solving the defined optimization problem heuristically. Our greedy algorithm iteratively i) adds channels and ii) selects fees such that each added channel increases the profit maximally for the previously selected channels. For this purpose, we leverage the concept of vertex and edge betweenness centrality, i.e., the fraction of cheapest paths a vertex or edge is contained in. 
We evaluate our algorithm for real-world data set based on the Lightning network. Our evaluation strongly indicates that our approach does not only greatly improve the profit in comparison to default fees but also that leveraging betweenness centrality for selecting channels offers considerably better results than other network centrality measures. 
More precisely, our algorithm increases the profit by a factor 4 in comparison to default fee values and is at least a factor 2 better than other strategies based on network centrality metrics. 
An important result is that the nodes with already established channels can still gain an advantage by using our fee selection algorithm to replace their default fee values.

\section{Background}\label{sec:prelim}

This section summarizes key concepts from the field of payment channels required to understand the remainder of the paper. Furthermore, as our algorithm relies on graph centrality metrics, this section defines these metrics and gives some intuition on their role.  

\subsection{Payment Channel Networks}
Payment channel networks (PCN) are one of the two key approaches to scaling blockchains by moving transactions off-chain~\cite{gudgeon2019sok}. 
Two parties open a payment channel through an initial funding transaction on the blockchain that locks coins such that they can only be used for transactions between the two parties. After this initial funding transaction, the two parties can conduct payments without directly interacting with the blockchain. They commit to the latest balance of the channel, i.e., the distribution of the total number of locked coins over the two parties. 
For instance, let nodes $u$ and $v$ open a payment channel such that $u$ locks $x$  coins and $v$ locks $y$ coins. The initial \emph{balance} of the channel is $(x,y)$ and its total \emph{capacity} is $x+y$. If $u$ sends one coin to $v$, the balance changes to $(x-1,y+1)$. 

In case of a dispute about the channel balance, the signed commitments documenting the state changes are published on the blockchain. The blockchain consensus then assigns the coins according to the latest valid channel state. Once the two parties decide to close their channel, they have to conduct a closing transaction on the blockchain. Afterward, they receive the coins locked in the channel with the number of coins per party corresponds to the channel balance at the time of the closure. 
In the absence of disputes, the intermediary transactions are almost instant and the number of transaction is merely bound locally by the bandwidth and latency of nodes. 

Establishing a payment channel does not make sense if parties do not trade with each other regularly due to i) the on-chain fees for establishing the channel and ii) the opportunity cost caused by locking coins to the channel. Thus, most nodes will only establish a few channels with frequent trading partners. Routing payments via a path consisting of multiple channels nevertheless allows nodes to trade without having a direct channel. For instance, a node $s$ can make a payment to a node $r$ via two intermediary nodes $u$ and $v$, meaning that the payment is routed via three payment channels: $s$ to $u$, $u$ to $v$, and $v$ to $r$. The balances along all these channels change according to the transaction value.  

The intermediary nodes charge fees for the use of their channels. For a channel $Ch_i$ from $u$ to $v$, these fees consist of a basic fee $BF_{Ch_i}$ for using the channel and fee rate $FR_{Ch_i}$ per transferred unit. The overall fee of a transaction $tx$ for the channel is hence
\begin{eqnarray}\label{eqn:lightningfee}
\fee{Ch_i,tx}= BF_{Ch_i} +FR_{Ch_i} \cdot |tx|,
\end{eqnarray}
where $|tx|$ denotes the transaction amount.
The fees are determined by and paid to $u$. The sender $s$ has to pay the fees. 
Note that the fee calculation formula given in Equation \ref{eqn:lightningfee} is specific for the Lightning network \cite{bolt}.
Still, the other payment or state channel networks have a similar structure.

%\todo{where do we have the info about the base fee + rate?, is that Lightning-specific?}

\subsection{Graph Centrality Metrics}
In this work, we model a PCN network as a directed graph.
In this manner, each node in the payment channel represents and vertex in the graph and each channel is represented by two directional edges between the nodes (one for each direction).
The channel fees correspond to the weights of the edges.

As a consequence, we can make use of graph metrics that characterize the importance of certain nodes in a weighted directed graph. Our key metrics are (vertex) betweenness centrality and edge betweenness centrality.  

\begin{definition}[Betweenness Centrality] 
	The betweenness centrality of a vertex~\cite{Freeman1977} is proportional the total number of shortest paths that pass through that vertex, i.e., 	
	\begin{eqnarray*}
		\bc{v} = \sum_{\substack{s\neq t \neq v \\ \sigma_{st}\neq 0 }} \frac{\sigma_{stv}}{\sigma_{st} },
	\end{eqnarray*}
	where $\sigma_{st}$ denotes the number of shortest paths between $s$ and $t$ and  $\sigma_{stv}$ is the number of such shortest paths containing the vertex $v$. 

Similarly, the edge betweenness centrality~ \cite{Girvan7821} of an edge relates to the total number of shortest paths that pass through that edge, i.e., 	
\begin{eqnarray*}
	\ebc{[v_1v_2]} = \sum_{\substack{s\neq t \\ \sigma_{st}\neq 0 }} \frac{\sigma_{st[v_1v_2]}}{\sigma_{st} },
\end{eqnarray*}
where $\sigma_{st[v_1v_2]}$ is the number of shortest paths passing through the edge $[v_1v_2]$.
\end{definition}

We make use of the following result about vertex betweenness centrality to assess the suitability of our greedy heuristic for selecting channel fees. 

\begin{theorem}[\cite{DBLP:journals/jea/BergaminiCDMSV18}]\label{thm:bc}
	For each vertex $v$, betweenness centrality function $\bc{v}$ is a monotone function for the set of edges incident to $n$.
\end{theorem}

An important problem concerning the betweenness centrality is the \textit{maximum betweenness improvement} ($\mathsf{MBI}$) problem.
\begin{definition}[MBI problem \cite{DBLP:journals/jea/BergaminiCDMSV18}] \label{def:mbi}
	For a given directed graph $G$ and a node $n$, finding a limited number of channels incident to node $n$ such that $\bc{n}$ would be maximized is called as the \emph{maximum betweenness improvement} problem.
\end{definition}
With the help of the following theorem on $\mathsf{MBI}$ problem, we prove that our problem of maximizing the reward ($\mathsf{MRI}$) is also NP-hard.
\begin{theorem}[\cite{DBLP:journals/jea/BergaminiCDMSV18}]\label{thm:MBIapprox}
	MBI problem cannot be approximated in polynomial time within a factor greater than $1-\frac{1}{2\epsilon}$,unless $P=NP$.
\end{theorem}

\section{Our PCN Model}\label{sec:model}
There are a number of PCNs with Lightning~\cite{lightning}, Raiden~\cite{raiden}, Perun \cite{dziembowski2017perun} and Celer \cite{dong2018celer} being key examples. All of them use slightly different assumptions and properties. We base our system model on Bitcoin's Lightning network.

In the following, we first describe our PCN model $\mathbf{LN}$. In this model, we then define the problem of an individual participant aiming to maximize their gain. 
We summarize the notation used in the paper in Table \ref{tab:not} in Section \ref{sec:notation}. 

\subsection{Network Topology, Fees, and Routing}
Nodes open and close payment channels via. blockchain transaction. For simplicity, we assume that the opening cost $ChCost$ remains constant over time.  

In Lightning, the complete topology of the network is known. Nodes publicly announce on the blockchain that they establish or close a channel. Furthermore, nodes willing to route payments are announcing their channels and fees to the complete network.
Thus, we assume in our model that both the topology and the fees of all nodes are publicly known. For simplicity, we assume that the topology and routing fees of the nodes that do not strategically change them remain fixed over time. Otherwise, our fee selection strategy would require a model to anticipate the expected changes. Current research on payment channel networks does not provide such a model. Our analysis on the Lightning network indicates that fees are indeed usually the default value. As topology changes require on-chain transactions, which are costly in both time and on-chain fees, the topology also should not change considerably.
Moreover, we assume that nodes apply source routing to find one cheapest path from source to destination, as is the case in the current implementation of Lightning. 

\subsection{Problem Definition}

We represent a network $\mathbf{LN}$ as a graph $G=(V,E)$ of $V$ vertices $V$ and edges $E$, we consider $\mathbf{A}$ that wants to maximize its revenue in running a node in a payment channel network. $\mathbf{A}$ opens channels with other nodes in the network, each channel having a total cost of $ChCost$ for opening and closing. We assume that $\mathbf{A}$ can strategically select the nodes it establishes channels with from all nodes in the network. 
After all, these nodes do not need to invest anything into the channel as $\mathbf{A}$ completely funds them and will likely receive additional monetary gains through routing fees. 
We assume that $\mathbf{A}$ has a budget of $\mathbf{c(A)}$ coins to use as collateral for the channels in total. 

Formally, let $C$ be the set of chosen channels. For each channel $Ch_i \in C$, we have the coins allocated to the channel $\coins{Ch_i}$ and the channel fee $\fee{ Ch_i, tx}$ for a transaction value $tx$. 
Wlog, transaction values are integers between 1 and $\mathbf{T}_{max}$ following a distribution $T$.
Let $X_i(tx, S, R)$ be the event that a transaction of value $tx$ going from a node $S$ to a node $R$ passes through the channel $Ch_i$. 
Then the expected fee from that transaction is $\fee{ Ch_i, tx }Pr[X_i(tx, S, R)]$. Last, we require the distribution $M$ that returns a sender-receiver pair.  
$\mathbf{A}$'s objective is to find $C$, $f$, and $\coins$ such that the overall expected gain of one transaction 
\begin{eqnarray}
\label{eq:objective}
 \sum_{\substack{\forall S, R \in V\\ S\neq R\neq \mathbf{A}} } Pr(M=(S,R))   \sum_{j=1 }^{\mathbf{T}_{max}} Pr(T=j) \sum_{Ch_i \in C} \fee{ Ch_i, j } \cdot Pr[X_i(j, S, R)]
\end{eqnarray}
is maximized while adhering to the constraint that $\sum_{Ch_i \in C}\coins{Ch_i} \leq \coins{ \mathbf{A} }$. 
%The second term in Eq.~\ref{eq:objective} is the total cost for opening and closing the channels. 
Equation ~\ref{eq:objective}  computes the expected gain over the involved variables $T$ and $M$.   
If the capacity of the channel $\coins{Ch_i}$ is less than the transaction amount $tx$,  $Pr[X_i(tx, S, R)]=0$.
Similarly, if there does not exist a shortest path from $S$ to $R$ that  passes through $Ch_i$, $Pr[X_i(tx, S, R)]=0$. 
Otherwise, $Pr[X_i(tx, S, R)]$ is equal to the number of shortest paths from $S$ to $R$ passing through $Ch_i$ divided by the total number of shortest paths from $S$ to $R$.

Note that Equation~\ref{eq:objective} ignores the cost of opening $C$ channels, $|C| \cdot ChCost$. The impact of this cost depends on the number of transactions $K$ that occur during the lifetime of a channel. Let $max$ be the maximal value for Equation~\ref{eq:objective}. The overall gain of the node is then the difference: $K\cdot max-|C| \cdot ChCost$.  By increasing the lifetime of the channel arbitrarily, the impact of $|C| \cdot ChCost$ diminishes, which is why we disregard it for Equation~\ref{eq:objective}. Our model furthermore disregards the opportunity cost caused by locking coins due to the absence of suitable models for such a cost.

\section{Our Fee Strategy}\label{sec:strategy}

%Our fee strategy consists of two parts
%\begin{enumerate}
%	\item \textit{Channel selection}: Party $\mathbf{A}$ decides with whom to open a channel and determines the fee for each channel maximizing the total reward.
%	\item \textit{Capacity distribution}: Once the channels are chosen, party $\mathbf{A}$ decides how to distribute $\coins{\mathbf{A}}$ coins amongst them.
%\end{enumerate}

We start by showing that maximizing the objective function given in Equation~\ref{eq:objective} is NP-hard. 
Afterward, we present our greedy algorithm for approximating a solution. As our algorithm contains an equation for choosing channel fees without a closed-form solution, the last part of the section demonstrates a method for solving the equation numerically.  

Our proof and algorithm act on a version of Equation~\ref{eq:objective} for specific distributions $T$ and $M$.
In the absence of real-world data for these distributions, we utilize two straight-forward distributions.  
Concretely, our work considers a fixed transaction value, i.e., the random variable $T$ only takes one value $tx$. 
For the distribution $M$, which characterizes the likelihood of two nodes to trade, assuming that all nodes are equally likely to trade with each other is the most natural choice in the absence of a concrete alternative model. Thus, $M$ is a uniform distribution over all pairs of nodes in the following.

For the design of our algorithm, we furthermore bound the maximal channel fee by $f_{max}$. Assuming a maximal channel fee does not reduce the generality of our approach. As nodes send payments along the path with the lowest fee, any channel fee that entails the channel is not contained in any such path can be disregarded. 

\subsection{NP-hardness of the Problem}
Before presenting the actual proof, we rephrase Equation~\ref{eq:objective} to relate it to the concept of (edge) betweenness centrality. 

Choosing $M$ to be a uniform distribution implies that $Pr(M=(S,R))=\frac{1}{(|V|-1)(|V|-2)}$\footnote{$(|V|-1)(|V|-2)$ is the number of pairs of nodes when not including $\mathbf{A}$} is a constant, which can disregarded for the optimization. Furthermore, choosing a constant transaction value $tx$ removes the second sum in Equation~\ref{eq:objective}.
Hence our modified objective function is 
\begin{eqnarray}
\label{eq:objectiveMod}
 \sum_{Ch_i \in C} \fee{ Ch_i, tx } \cdot Pr[X_i(tx, S, R)].
\end{eqnarray}

The next step relates $Pr[X_i($tx$, S, R)]$ in Equation~\ref{eq:objectiveMod} to the betweenness centrality. There are two important quantities to consider: the number of shortest paths including the channel and total fee reward gained from these paths.
Maximizing the number of shortest paths passing through a channel or node corresponds to the edge or vertex betweenness centrality (BC), respectively, as defined in Section~\ref{sec:prelim}. However, maximizing the BC does not necessarily imply maximal revenue. As fees represent edge weights, the shortest path here is a path whose edges have the minimal sum of weights. Choosing low fees hence increases the probability to be contained in the shortest path but low fees also indicate a low gain from each transaction. 

Rather, the expected reward of a channel $Ch_i$ is equal to the probability of the transaction passing through that channel times the fee.  Note that each channel needs to have a capacity of at least $tx$ for the payment to choose this path,
thus a optimal solution for Equation~\ref{eq:objectiveMod} will only create channels of sufficient capacity and we can exclude the capacity aspect from $Pr[X_i(tx, S, R)]$.
With $\ebc{Ch_i} $ denoting the edge betweenness centrality of a channel $Ch_i$ with fees $\fee{Ch_i}$\footnote{For the rest of section, we drop the transaction amount $tx$ from the channel fee formula $\fee{Ch_i}$ as it is fixed.}, the formal expression for the expected reward of $Ch_i$ is 
\begin{eqnarray}\label{eqn:er}
\mathsf{ER} (Ch_i) =  \fee{Ch_i} \cdot \ebc{Ch_i}.
\end{eqnarray}
%Note that, from now on, we drop the transaction amount $tx$ from the channel fee formula $\fee{Ch_i}$ as it is fixed.
As a consequence, the total expected reward of $\mathbf{A}$ from Equation~\ref{eq:objectiveMod} is 
\begin{eqnarray}\label{eqn:er_node}
	\mathsf{ER}(\mathbf{A}) =  \sum_{ \forall Ch_i \in C } \mathsf{ER}(Ch_i).
\end{eqnarray}
Now, we can formally define the problem from Equation~\ref{eq:objective} as the \textit{maximum reward improvement} (MRI) problem.

\begin{definition}[MRI Problem]
For a payment channel network $\mathcal{LN}$ and a node $n$, the problem of finding a set of channels $C$ such that $\mathsf{ER}(n)$ is maximized is called the \emph{maximum reward improvement} problem.
\end{definition}

The following theorem states that it is not possible to design an algorithm $\f{cs}$ that finds the optimum solution within polynomial time, unless $P=NP$.

\begin{theorem}[MRI Approximation Theorem]\label{thm:our}
$\mathsf{MRI}$ problem cannot be approximated in polynomial time within a factor greater than $1-\frac{1}{2\epsilon}$, unless $P=NP$.
\end{theorem}
\begin{proof}
	To prove this theorem, we reduce our $\mathsf{MRI}$ problem to $\mathsf{MBI}$ problem presented in Definition \ref{def:mbi}.
	Using Equation \ref{eqn:er_node}, we can formulate the $\mathsf{MRI}$ problem as follows:
	\begin{eqnarray*}
	\mathsf{MRI}(\mathbf{LN},n,N_c) \rightarrow
	\mathcal{CH}_M = \argmax_{ \substack{ \forall \; |\mathcal{CH}| \leq N_c \\ \sender{Ch_i}=n \\ \forall \; \fee{Ch_i} \in [1,f_{max}] } } \left( \mathsf{ER}(n) = \sum_{ \forall Ch_i \in \mathcal{CH} } \mathsf{ER}(Ch_i) \right).
	\end{eqnarray*}
	
	We introduce a subproblem, namely $\mathsf{MRI\_FF}$, where the upper limit of the fee $f_{max}$ is equal to 1, which means that all the channel fee are equal to 1.
	Using the Equation \ref{eqn:er}, $\mathsf{MRI\_FF}$ can be formulated as:
	
		\begin{eqnarray}\label{eqn:reduction}
		&&\mathsf{MRI\_FF}(\mathbf{LN},n,N_c) \rightarrow
		\mathcal{CH}_M = \argmax_{ \substack{ \forall \; |\mathcal{CH}| \leq N_c \\ \sender{Ch_i}=n } } \left( \sum_{ \forall Ch_i \in \mathcal{CH} }   \ebc{Ch_i} \right) \\ \nonumber
		&&
		\stackrel{(*)}{=}\argmax_{ \substack{ \forall \; |\mathcal{CH}| \leq N_c \\ \sender{Ch_i} || \receiver{Ch_i} =n  } } \left(  bc_{n}  \right) \stackrel{(**)}{=}\mathsf{MBI}(\mathbf{LN},n,N_c),
	\end{eqnarray}
	which reduces to the $\mathsf{MBI}$ problem.
    Here, the first equality $(*)$ holds because the summation of the all shortest paths passing from out-going edges is equal to the total number of shortest paths passing through that node.
    In other words, the summation of edge betweenness centrality of all out-going edges of a node is equal to betweenness centrality of that node.
    The second equality $(*)$ follows from the definition of the $\mathsf{MBI}$ problem given in Definition \ref{def:mbi}.
    
    Now, we can prove our theorem by contradiction.
    Let assume there exists an approximation to $\mathsf{MRI}$ problem within a factor greater than $1-\frac{1}{2\epsilon}$.
    Then, the same approximation would hold for the subproblem of $\mathsf{MRI}$, $\mathsf{MRI\_FF}$ with a certain maximal fee, namely 1. However, in Equation \ref{eqn:reduction}, we showed that $\mathsf{MRI\_FF}$ problem is equivalent to the $\mathsf{MBI}$ problem.
    This contradicts Theorem \ref{thm:MBIapprox}.
    Therefore, $\mathsf{MRI}$ problem cannot be polynomially approximated within a factor greater than $1-\frac{1}{2\epsilon}$, unless $P=NP$.
\end{proof}

\subsection{Channel Selection Function}
We present a \textit{greedy} algorithm $\f{cs}$ to approximate the $\mathsf{MRI}$ problem.
$\f{cs}$ takes the PCN and the requested number of channels as input and outputs the set of nodes to whom channels are created.
It internally calls $\f{cf}$, the algorithm for deciding the fee of a channel.
Formally, we have 
\begin{eqnarray}\label{eqn:cff}
	&&\f{cf}(\mathcal{CH} \cup Ch) \rightarrow R_{Ch}:\nonumber\\ 
	&& R_{Ch}= \mathsf{TotalER}(\mathcal{CH} \cup Ch,f) \text{ where } f =\argmax_{f_i \in [1,f_{max}] } \left( \mathsf{TotalER}(\mathcal{CH} \cup Ch,f_i)  \right),\nonumber\\
	&& \mathsf{TotalER}(\mathcal{CH} \cup Ch,f_i)= \mathsf{ER} (Ch)_{\fee{Ch}=f_i}+\sum_{ \substack{ \forall Ch_j \in \mathcal{CH} } } \mathsf{ER} (Ch_j).  
\end{eqnarray}

As detailed in Algorithm \ref{alg:csf}, our greedy algorithm for $\f{cs}$ consists of the following five key steps:
\begin{enumerate}
	\item Start with an initial PCN of nodes and channels.
	\item At each step, try all possible channels between our node and other nodes. 
	\item Compute the maximum reward of the channel by using $\f{cf}$.
	\item Connect to the node who gives the maximum reward and update the PCN.
	\item Go to step (2) until the desired number of channels is established. 
\end{enumerate}

\begin{algorithm}
	\caption{Channel Selection Function}
	\label{alg:csf}
	\begin{algorithmic}[1]
		\Require $\mathbf{LN}$ and $N_c$
		\Ensure $\mathcal{CH}$
		\Function{$\f{cs}$}{$\mathbf{LN},N_c$} 
		\State $ \mathcal{CH} \leftarrow \emptyset$
		\While{$|\mathcal{CH}|<N_c$ }
		\State $ maxRew \leftarrow 0, selectednode = None$
		\For{Each node $n_i \in \mathbf{LN}$} 
		\State Create a channel between $(n,n_i)$: $\mathbf{LN}_i\leftarrow AddEdges(\mathbf{LN},[n,n_i])$
		\State Calculate the reward $R_{n_i} \leftarrow \f{cf}(\mathbf{LN}_i,\mathcal{CH} \cup [n,n_i])$
		\If{$maxRew \leq  R_{n_i}$}
		\State $maxRew = R_{n_i}$
		\State $selectednode=n_i$
		\EndIf
		\EndFor
		\State $\mathcal{CH} \leftarrow \mathcal{CH} \bigcup \{selectednode\}$
		\State $\mathbf{LN} \leftarrow AddEdges(\mathbf{LN},[n,selectednode])$
		\EndWhile
		\State \Return $\mathcal{CH}$
		\EndFunction
	\end{algorithmic}
\end{algorithm}

Next, we ascertain that channel additions cannot reduce the expected revenue, indicating that nodes should add all channels they can fund. 
Here, it is important to note that we do not take into account the channel opening cost $ChCost$.
Thus, if the marginal reward improvement of a new channel is zero, there is no point in add the channel.

\begin{theorem}[Monotonicity]\label{thm:ourgreedy}
The objective function of Algorithm \ref{alg:csf} is a monotone non-decreasing function.
\end{theorem}
\begin{proof}
%First, we need to show that our objective function $\f{cf}$ is monotone. 
A function $\mathcal{F}: \Omega \rightarrow \mathbb{R} $ is a monotone function if it satisfies the following condition:
	\begin{eqnarray}
	\forall S \subseteq T \subseteq \Omega, \quad  \mathcal{F}(S) \leq  \mathcal{F}(T). 
	\end{eqnarray} 	
In our case, for any solution $\mathcal{CH}$ to $\mathsf{MRI}$ and for any node $n_i$ such that $[n,n_i] \notin \mathcal{CH}$, the following inequality holds $\f{cf}(\mathcal{CH} \cup [n,n_i]) \geq \f{cf}(\mathcal{CH})$.
Note that $\f{cf}$ checks for all possible fee values to maximize the total reward. 
In that sense, it would be enough to show that for the maximum fee value $f_{max}$, which can be formulated by using Equation \ref{eqn:cff} (with $\mathbf{LN}_0= \mathbf{LN}\cup \mathcal{CH}$ and $\mathbf{LN}_i= \mathbf{LN}\cup \mathcal{CH} \cup [n,n_i]$):

	\begin{eqnarray*}
		&&\f{cf}(\mathbf{LN},\mathcal{CH} \cup [n,n_i])  \geq \mathsf{TotalER}(\mathbf{LN},\mathcal{CH} \cup [n,n_i],f=f_{max})   \stackrel{?}{\geq} \f{cf}(\mathbf{LN},\mathcal{CH} ) \iff \\
		&& \mathsf{ER} (Ch,\mathbf{LN}_i )_{\fee{Ch}=f_{max}}+\sum_{ \substack{ \forall Ch_j \in \mathcal{CH} } } \mathsf{ER} (Ch_j,\mathbf{LN}_i ) 
		\stackrel{?}{\geq} \sum_{ \substack{ \forall Ch_j \in \mathcal{CH} } } \mathsf{ER} (Ch_j , \mathbf{LN}_0 ) \iff  \\
		&& \ebc{[n,n_i], \mathbf{LN}_i } \cdot f_{max} +  \sum_{ \substack{ \forall Ch_j \in \mathcal{CH} } } \ebc{Ch_j, \mathbf{LN}_i}   \cdot \fee{Ch_j} \stackrel{?}{\geq} 
		\sum_{ \substack{ \forall Ch_j \in \mathcal{CH} } }  \ebc{Ch_j,\mathbf{LN}_0 }  \cdot \fee{Ch_j} \\
		&& \iff \ebc{[n,n_i], \mathbf{LN}_i } \cdot f_{max} \stackrel{?}{\geq}   \sum_{ \substack{ \forall Ch_j \in \mathcal{CH} } } \left( \ebc{Ch_j, \mathbf{LN}_0}  -   \ebc{Ch_j, \mathbf{LN}_i}  \right) \cdot \fee{Ch_j} \\
		&& \stackrel{(*)}{\impliedby}\ebc{[n,n_i], \mathbf{LN}_i } \stackrel{?}{\geq}   \sum_{ \substack{ \forall Ch_j \in \mathcal{CH} } } \left(  \ebc{Ch_j, \mathbf{LN}_0} - \ebc{Ch_j, \mathbf{LN}_i} \right) \\
		&& \iff  \ebc{[n,n_i], \mathbf{LN}_i } + \sum_{ \substack{ \forall Ch_j \in \mathcal{CH} } }   \ebc{Ch_j, \mathbf{LN}_i}  \stackrel{?}{\geq}  \sum_{ \substack{ \forall Ch_j \in \mathcal{CH} } }   \ebc{Ch_j, \mathbf{LN}_0}\\
		&& \stackrel{(**)}{\iff}  \bc{n, \mathbf{LN}_i}    \stackrel{?}{\geq}    \bc{n, \mathbf{LN}_0}  .
	\end{eqnarray*}
%}
 
%Here, $(*)$ comes from the fact that the total edge betweenness centrality of a graph is fixed (for a fixed number of nodes).
Here, $(*)$ condition is true since for all channels $\fee{Ch_i}\leq f_{max}$ by the definition. 
$(**)$ is satisfied since the summation of edge betweenness centrality of all out-going edges of a node is equal to betweenness centrality of that node.
At the end, $ \bc{n, \mathbf{LN}_i}   \geq  \bc{n, \mathbf{LN}_0} $ holds because betweenness centrality is a monotone function, see Theorem \ref{thm:bc}. 

%\begin{eqnarray*}
%&&	\mathsf{ER} (Ch ,\mathbf{LN}\cup [n,n_i] )_{\fee{Ch}=f_{max}}+\sum_{ \substack{ \forall Ch_j \in \mathcal{CH} } } \mathsf{ER} (Ch_j,\mathbf{LN}\cup [n,n_i] ) \stackrel{?}{\geq} \sum_{ \substack{ \forall Ch_j \in \mathcal{CH} } } \mathsf{ER} (Ch_j , \mathbf{LN} ) \\ 
%&&	\implies \f{cf}(\mathbf{LN},\mathcal{CH} \cup [n,n_i])  \stackrel{?}{\geq} \f{cf}(\mathbf{LN},\mathcal{CH} ) 
%\end{eqnarray*}
%\todo{Show that it is submodular!!!!!! WORKING ON IT}
\end{proof}

%For that reason, in Algorithm \ref{alg:csf} of our $\f{cs}$, we also utilize the greedy algorithm. 

%Here, $\f{cf}$ is called to decide on the fee for each channel.

\subsection{Efficient Search Algorithm for the Channel Fee Function}
No closed-form formula finds the best fee amount maximizing the expected reward due to the term $\ebc{Ch}$ for a channel $Ch$.
Here, we analyze Equation \ref{eqn:er} to minimize the computational cost by discarding some parts of the search space.
First of all, since $\ebc{\mathbf{LN}}$ is not affected by the changes in the fees of channels, the denominator is irrelevant for optimizing the $\mathsf{ER}(Ch)$.
Therefore, $\f{cf}$ can be seen as function of the EBC of the channel $\ebc{Ch}$ and the fee $\fee{Ch}$.
Secondly, $\ebc{Ch}$ is negatively affected by $\fee{Ch}$ because increasing the fee means an increase in the weight of the edge that results in a lower chance of being in the shortest paths (see Figure \ref{fig:ebc} in Section \ref{sec:app_figure} for an illustrative example). 

Two observations give rise to an efficient search algorithm for finding the most suitable fee. 
The first observation utilizes the fact that edge betweenness centrality is a monotone decreasing function concerning the channel fee.  
Let the expected reward of a channel for chosen fees $f_3>f_1$ be $r_1=e_1\cdot f_1$ and $r_3=e_3\cdot f_3$, respectively. 
If $r_3> r_1$, let
%$f_2$ can be found by the following formula:
\begin{eqnarray}
f_2= f_1 \cdot \frac{r_3}{r_1} = f_3 \cdot \frac{e_3}{e_1}.
\end{eqnarray}
It can be seen that the expected reward for any fee $f_\alpha$, namely $r_\alpha$, where $f_1<f_\alpha\leq f_2$ is at most $r_3$:
\begin{eqnarray}
r_\alpha=e_\alpha \cdot f_\alpha \leq e_1 \cdot f_\alpha \leq e_1 \cdot f_2  = e_3 \cdot f_3 = r_3.
\end{eqnarray}
In other words, there is no need to compute the expected reward values for the fees in between $f_1$ and $f_2$ as they cannot be optimal values.

The second observation is that increasing the fee of an out-going channel $Ch$ cannot decrease the edge betweenness of another out-going channel $Ch'$ of the same node.
Such an increase can only reduce the edge betweenness of channels that are on a path containing $Ch$ by removing the path from the set of shortest paths. 
However, as shortest paths cannot have loops, two out-going channels of the same node cannot be on the same shortest path.
Now, let $\mathcal{CH}$ be the set of previously selected channels. 
Let $r'_1$ and $r'_3$ be the sum of the expected fees of all channels $Ch' \in \mathcal{CH}$ for fees $f_1$ and $f_3$ with $f_3 > f_1$.  By the above observation, we have $r'_3 \geq r'_1$.

Our recursive algorithm divides the space of all possible fee values from $1$ to $f_{max}$ into $\mathsf{d}$ intervals.
For each interval $i$, let $r_i=\mathsf{ER} (Ch,\fee{Ch}=f_i)$ be the expected reward  of $Ch$, $r_i' \leftarrow \sum_{ \substack{ \forall Ch' \in \mathcal{CH} } } \mathsf{ER} (Ch',\fee{Ch'})$ be the total reward of the other channels. 
By the first observation, the maximal increase in $r_i$ is $\frac{ f_{i+1} }{f_i}$ and by the second observation $r'_{i+1} \geq r'_i$ as $f_{i+1}> f_i$.
Thus, the maximum possible reward value for interval $i$ is $\widetilde{R}_i=r_i \cdot \frac{ f_{i+1} }{f_i}+r'_{i+1}$.
If $\widetilde{R}_i$ is greater than the current maximum reward value, the algorithm recursively searches for a maximum in the interval, otherwise discards the interval.
We present the pseudocode for the algorithm in Section~\ref{sec:pseudoCFF}.

This completes the description of our algorithm, which we evaluate in the following in comparison to other approaches based on common centrality metrics.

\section{Evaluation}\label{sec:implementation}

In this section, we evaluate our proposed fee strategy for a real-world topology. 
Our evaluation quantifies the total reward gained by $\mathbf{A}$ when using our greedy algorithm.

To emphasize the high effectiveness of our solution, we compared it with other channel and fee selection algorithm.
For the channel selection, we considered random nodes as well as connecting to nodes with a high centrality for three centrality metrics: i) degree, i.e., connecting to the nodes with the most connections, ii) betweenness centrality, and iii) pagerank~\cite{page1999pagerank}. 
For the fee strategy, we compute the results for both cases where the channel fees are the default values and they are determined by $\f{cf}$.

\subsection{Model}
In Lightning network, the upcoming transactions and current balances of channels are not known. Thus, we need to model the network and transactions.

\paragraph{Transactions.}
Like Section~\ref{sec:strategy}, our evaluation assumes that all source-destination pairs are equally likely. Furthermore, we categorize the transactions into three groups based on the amounts:
\begin{itemize}
	\item \textit{Micro payments} are the transactions involving a very small amount of coins. To represent this category, we use the transaction amount of 100 Satoshi, which is about one cent\footnote{https://awebanalysis.com/en/convert-satoshi-to-euro-eur/}. 
	An example of a use case would be the streaming services where you pay small amounts per service.
	\item \textit{Medium payments}: are the transactions spent for daily living expenses like buying a coffee, which is represented with 10000 Satoshi. 
	\item \textit{Macro payments}: are transactions of high amounts, which is represented with 1000000 Satoshi. The amount of these transactions are in the order of 100 Euros.
\end{itemize}

From these categories, it is most likely that micro payments are usually restricted to nodes that have a direct channel. Otherwise, the base fee for the payment greatly exceeds the actual payment value. 
Therefore, our target transactions are medium and macro payments, which are analyzed separately.
 
\paragraph{Network.}
Following our system model in Section~\ref{sec:model}, networks are represented as weighted directed graphs. 
The weights of the edges in the graph model are calculated according to the fee rate and base fees of the channels.
Since the fee rate depends on the transaction amount, the weights of the same edges for medium and macro payments will be different.
The graph generated for the medium (macro) payments is called medium (macro) graph.

\subsection{Setup}

We obtained a snapshot of the Lightning Network (LN) data from $1ml.com$ on July 10 2019, which contains $4618$ nodes and $68729$ edges in total.
When we delete the edges with insufficient capacity, the medium graph has $68697$ edges and the macro graph has $32193$ edges. 

As a node requires at least two connections to be contained in any shortest paths, we first connected $\mathbf{A}$ to the two nodes with the highest degree (, which happen to have the highest pagerank as well).
For these two connections, we use the default fee rate and base fee values in both directions of the edges.
%Thereby, the weight of the edge is computed with respect to the transaction amount.\stefc{unsure what the previous sentence means}
Based on this initial scenario, we now connect $\mathbf{A}$ to additional nodes. 

The experiments use $ChCost=8192$ Satoshi, which reflects the fluctuating Bitcoin transaction fee estimates\footnote{https://bitcoinfees.info/}.
When establishing a new channel, our simulation added edges in both directions. The base fee and the fee rate of the in-coming edge correspond to the default value to model that i) most users currently stick to the default values and ii) $\mathbf{A}$ has no control over the in-coming channel fees as they are determined by the other party. 
For the outgoing edges, we utilize either $\f{cf}$ to determine the best fee value or use default values. 
When using $\f{cf}$, we restrict the $f_{max}=ChCost$.
Otherwise, the total fee cost of the transaction in the payment network is higher than the cost in the Bitcoin network and the sender is hence unlikely to proceed with the payment. 

\subsection{Experimental Results}
%Figure \ref{fig:reward1} show the performance of our greedy algorithm on the improvement of total fee reward value.

Figures~\ref{fig:reward1} and~\ref{fig:reward2} show the performance of our greedy algorithm in comparison to the other approaches in terms of the total reward improvement per new channel connections. 
The x-axis shows the number of connections added and the y-axis represents the total reward of node $\mathbf{A}$.
Since, for each case, we start with the same two connections, the total reward values have the same offset. 
%To see a clearer impact of our fee strategy, we present all results but our greedy ones.

Figure~\ref{fig:reward1} displays the result for the medium graph. When using default values, the reward is consistently lower than for our fee selection algorithm. More precisely, for centrality-based selection of channels, fee optimization increases the reward by a factor of roughly 2. Selecting channels strategically doubles the gain further in comparison to using  Pagerank centrality, which is the most beneficial one of the centrality-based selection methods.  
Figure \ref{fig:reward2} shows the results of macro graph. The results are similar to the case of medium payments, though the overall gain is slightly higher. 

In terms of fee computation efficiency, our experimental results show that the recursive algorithm is given in Algorithm \ref{alg:csf} reduces the search space of fees in the magnitude of 10--100.

\begin{figure}[hptb!]
	\centering
	\includegraphics[width=0.75\textwidth]{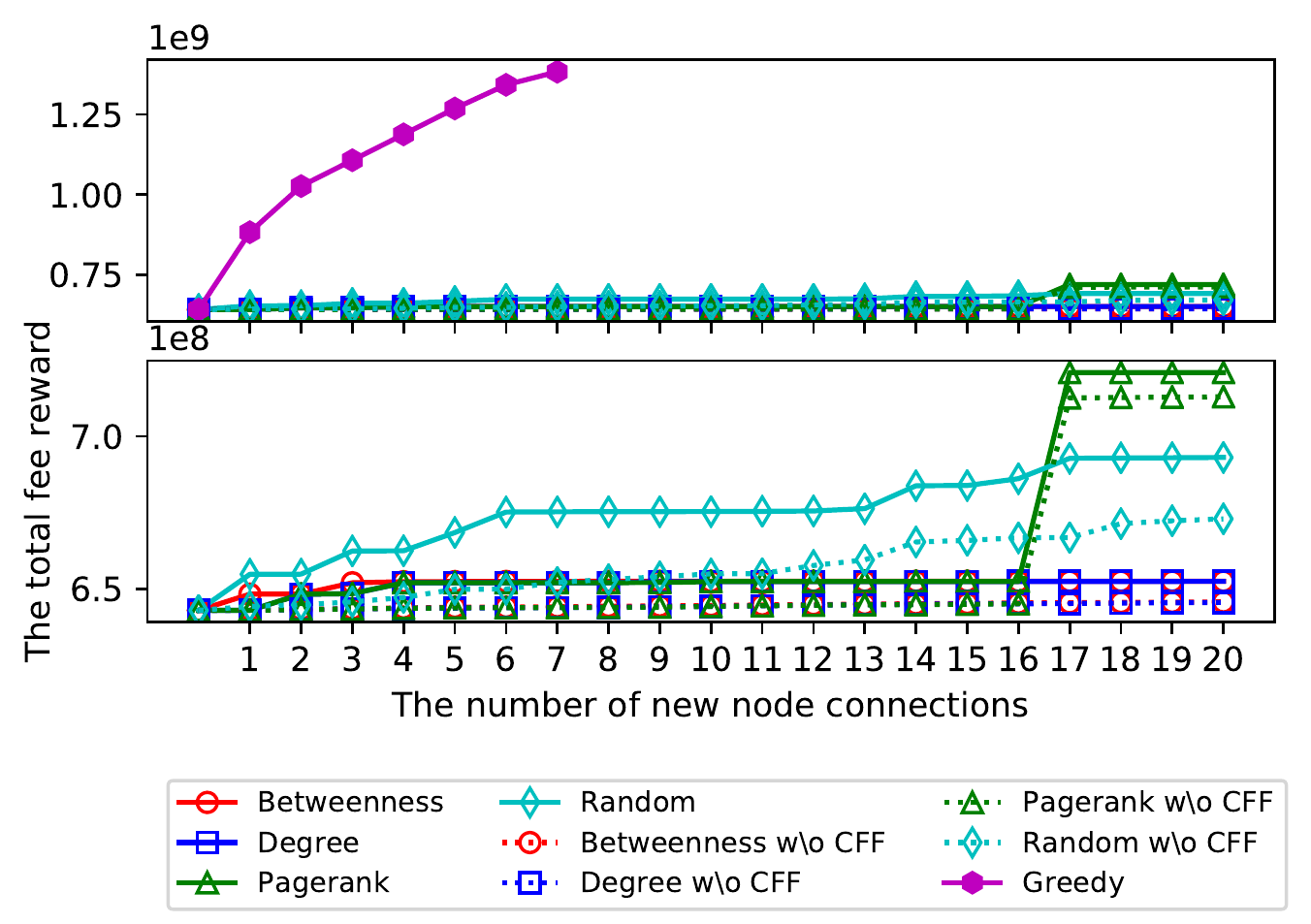}
	\caption{Total fee reward of our node in medium graph. The bottom figure excludes the greedy results to present a clear comparison of the rest.}
	\label{fig:reward1}
\end{figure}

\begin{figure}[hptb!]
	\centering
	\includegraphics[width=0.75\textwidth]{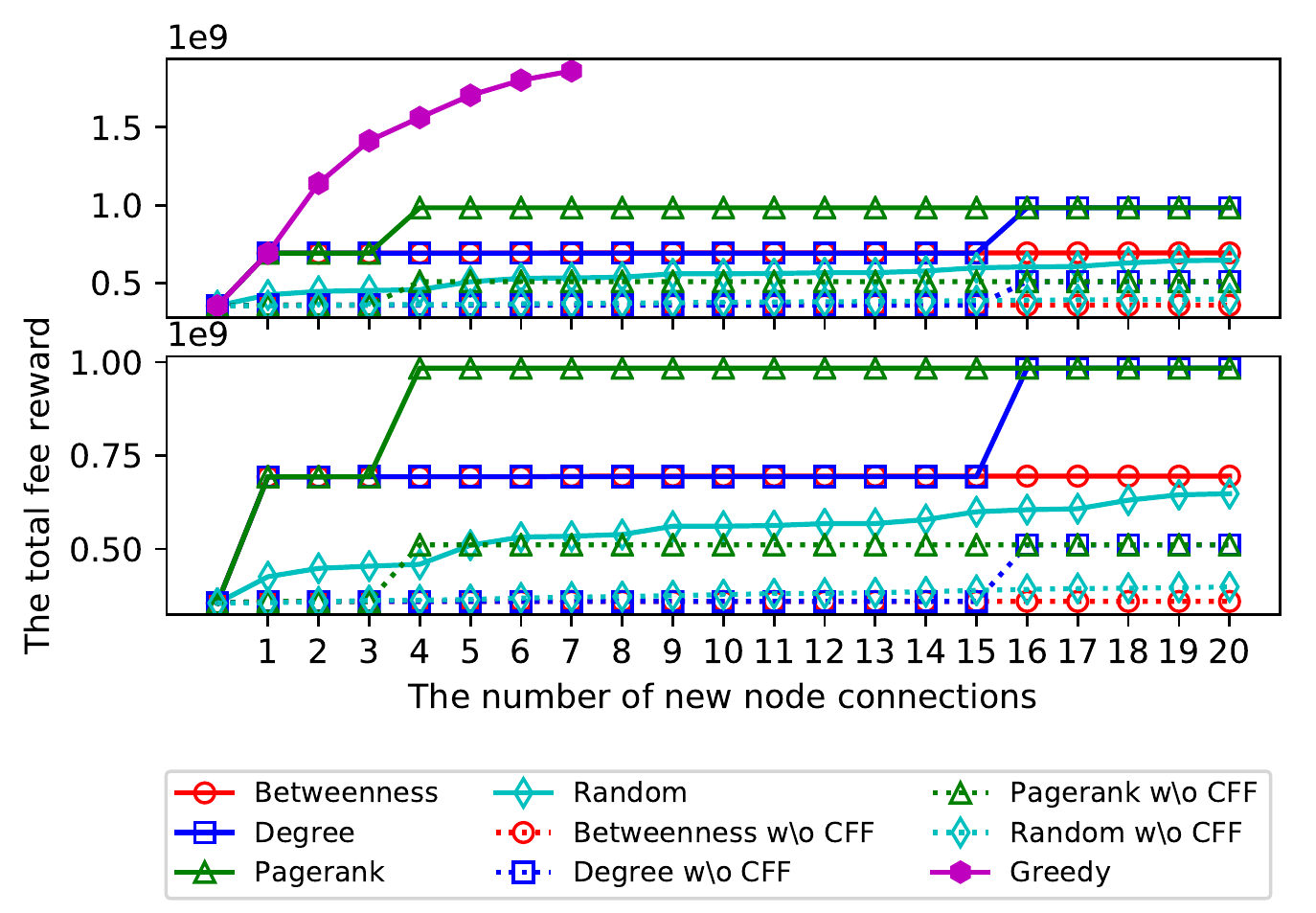}
	\caption{Total fee reward of our node in macro graph. The bottom figure excludes the greedy results to present a clear comparison of the rest.}
	\label{fig:reward2}
\end{figure}

%\begin{figure}[hptb!]
%	%\centering
%	\includegraphics[scale=0.8]{figures/MidTx_10000_greedy_result.pdf}
%	\caption{Betweenness centrality value of our node for medium cost transactions ($|Tx|=10000$ satoshi).}
%	\label{fig:reward2}
%\end{figure}

\subsection{Discussion}
From the experimental results, it can be seen that our greedy algorithm outperforms the centrality metrics.
Furthermore, the beneficial effect of the fee selection function can be observed by comparing the results with and without it.

Note that adding new connections to the nodes with the highest centrality metrics does not increase the total reward in comparison to random selection much, in particular for betweenness centrality. The reason here is that connecting to nodes with many shortest path passing them does not imply that the newly added channel offers a shorter path. Instead, directly focusing on the betweenness centrality of $\mathbf{A}$ results in larger improvements. 

Figure~\ref{fig:reward1} and ~\ref{fig:reward2} furthermore show few but notable differences between medium and macro payments. First, the overall gain is higher for macro payments as expected due to the higher transaction value and hence increased revenue for a similar fee rate. However, the base rate, which is 1000 Satoshi by default\footnote{The default fee values may change regarding the imported implementation. Our analysis on dataset shows that 33177 out of 68733 edges use the defaults we adopted.} in comparison to a default rate of 0.001, dominates the fee value, so that the 100-fold increase in the transaction value does not translate to a similar increase in gain. 
Secondly, the differences between various centrality measures are more distinct for macro payments, see Figure~\ref{fig:reward2}.

Overall, our greedy algorithm promises higher fees for individual nodes. Even if nodes cannot or do not desire to select their channels, they can still gain an advantage by using our more sophisticated fee selection algorithm for already established channels. 

One key limitation of our design is that it does not consider channel capacities as such. When all transactions have the same known value, $\mathbf{A}$ will only establish channels with sufficient collateral. However, in practice, $\mathbf{A}$ does not have such information and routing may fail due to a lack of capacity. Thus, integrating capacity information into both our model and our evaluation is clearly necessary in the future.

\section{Conclusion}

In this paper, we formalized an optimization problem for maximizing fees in payment channel networks, presented a heuristic algorithm for solving the problem, and evaluated our algorithm on real-world data sets. 
Our work demonstrates that routing fees can be a strong incentive for locking coins in payment channels. Fees as incentive hence have the potential to motivate rational users to fund payment channel and hence increase the ability of these networks to route payments. 

In this work, we focused on one individual node aiming to optimize its profit. Future work should design a game-theoretical framework for networks containing only rational nodes aiming to maximize their profit. For the continued usage of payment channel networks, incentives should ensure that strategies for optimizing profit locally also optimize the overall network health in terms of the availability of cost-effective paths. It remains an open question if the current fee model is a suitable incentive to further collaboration and network health.

\bibliographystyle{splncs03}
\bibliography{references}

\appendix
\section{Notation}\label{sec:notation}
The list of abbreviations and notations are given in Table \ref{tab:not}. 

\begin{table}[htp]\caption{Notation and Abbreviation Table}\label{tab:not}
	\begin{center}
		\begin{tabular}{ll}
			\hline\noalign{\smallskip}
			Symbol&Explanation\\
			\noalign{\smallskip}\hline\noalign{\smallskip}
			%\rule[-1ex]{0pt}{3.5ex}            $\f{nf}$& The network fee function.\\ 
			\rule[-1ex]{0pt}{3.5ex}            $\f{cs}$& The channel selection function.\\
			\rule[-1ex]{0pt}{3.5ex}            $\f{cf}$& The channel fee function.\\ 
			\rule[-1ex]{0pt}{3.5ex}            $\f{cd}$& The capacity distribution function.\\ 
		%	\rule[-1ex]{0pt}{3.5ex}            $\f{tf}$& The transaction fee function.\\ 
		%	\rule[-1ex]{0pt}{3.5ex}            $\f{ts}$& The transaction selection function.\\ 
			%	\rule[-1ex]{0pt}{3.5ex}  & \\
		%	\rule[-1ex]{0pt}{3.5ex}            $\mathbb{TX}$& The transaction set.\\ 
			\rule[-1ex]{0pt}{3.5ex}            $\mathbf{LN}$& The payment network.\\ 
			%	\rule[-1ex]{0pt}{3.5ex}  & \\
			\rule[-1ex]{0pt}{3.5ex}            $\coins{X}$& The total amount of coins of $X$.\\
			%\rule[-1ex]{0pt}{3.5ex}            $Ch_i$& The channel $i$.\\
			\rule[-1ex]{0pt}{3.5ex}            $\fee{Ch_i}$& The charging fee of the channel $Ch_i$.\\
			\rule[-1ex]{0pt}{3.5ex}            $\sender{Ch_i}$, $\receiver{Ch_i}$& The source and destination nodes of the channel $Ch_i$.\\
			%	\rule[-1ex]{0pt}{3.5ex}            $c_i$& The capacity of channel $i$.\\
			\rule[-1ex]{0pt}{3.5ex}            $\sender{tx}$, $\receiver{tx}$& The sender and receiver of the transaction $tx$.\\
			%	\rule[-1ex]{0pt}{3.5ex}            $t_{o,i}$, $t_{c,i}$& The channel $i$ opening and closing time.\\
			\rule[-1ex]{0pt}{3.5ex}            $ChCost$& The channel opening and closing on-chain cost.\\
			%			\rule[-1ex]{0pt}{3.5ex}            $ChCl$& The channel closing cost at time $t$.\\
			\hline\noalign{\smallskip}
		\end{tabular}    
	\end{center}    
\end{table}

\section{Illustrative example of the EBC vs. fee relationship of a channel}\label{sec:app_figure}

\begin{figure}[h]
	\centering
	\includegraphics[scale=0.5]{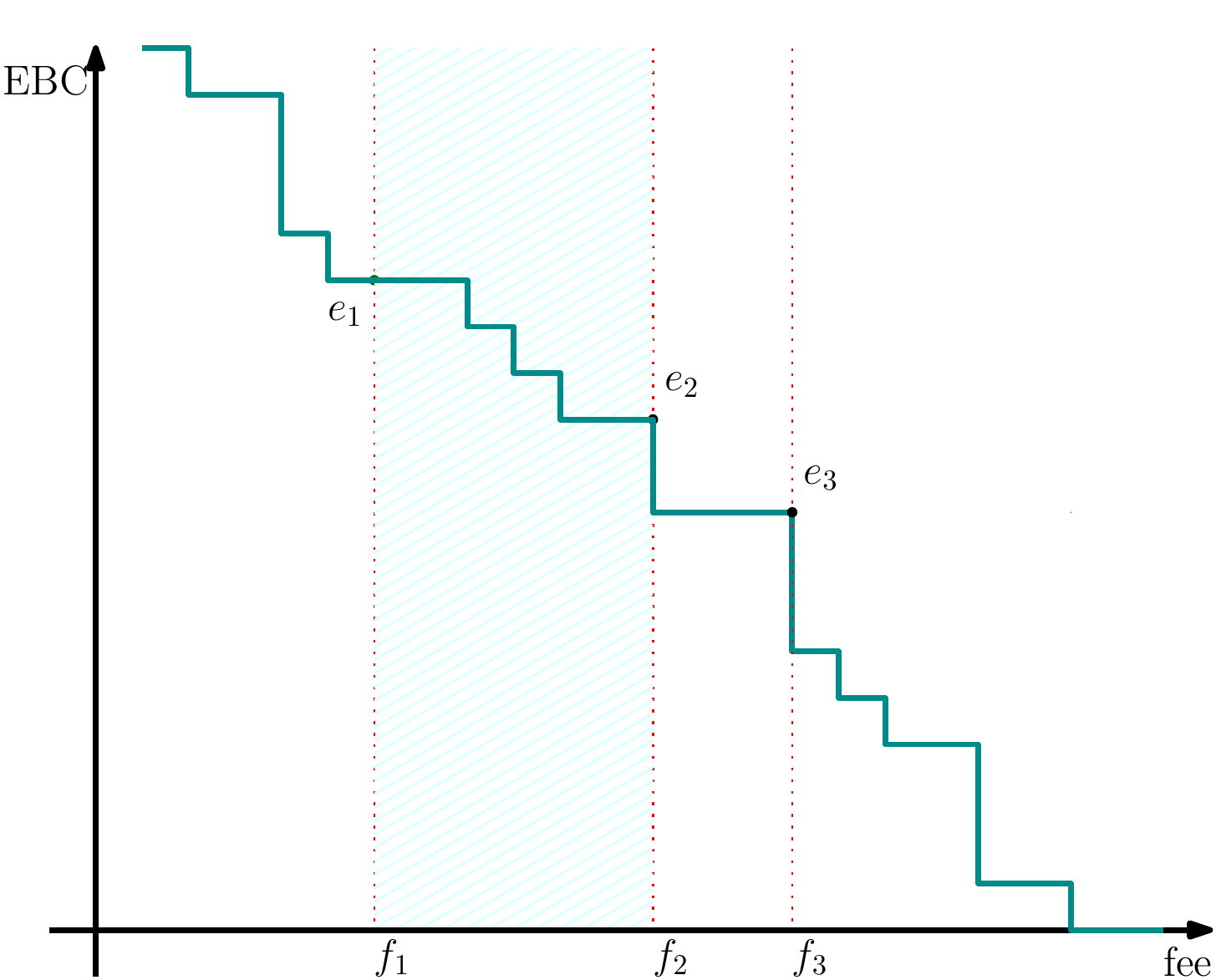}
	\caption{Illustrative example of the EBC vs. fee relationship of a channel.}
	\label{fig:ebc}
\end{figure}

\section{Pseudocode Channel Fee Function}
\label{sec:pseudoCFF}
Algorithm~\ref{alg:cff} is a recursive algorithm for determining the best fee in one step of the greedy algorithm. 

\begin{algorithm}
	\caption{Channel Fee Function}
	\label{alg:cff}
	\begin{algorithmic}[1]
		\Require $\mathbf{LN}$, $\mathcal{CH}$ and $Ch$
		\Ensure $R_{max}$ and $f_{max}$
		\Function{$\f{cf}$}{$\mathbf{LN},\mathcal{CH} \cup Ch,f_l,f_h$} 
		\State \% Initialization: $f_l \leftarrow 1, f_h \leftarrow ChCost,R_{max} \leftarrow 0,f_{max} \leftarrow 1$
		\State \% $\mathsf{d}$ is the division parameter
		\If{$f_h - f_l \leq \mathsf{d}$} \% Anchor step:
		\For{$f \in \{ f_l,\ldots,f_h \}$}
		\State $[r,r'] \leftarrow \mathsf{TotalER}(\mathcal{CH} \cup Ch,f)$
		\State Calculate the reward $R \leftarrow  r+r' $
		\If{$R\geq R_{max}$}
		\State $R_{max}\leftarrow R$ 
		\State $f_{max} \leftarrow f$
		\EndIf
		\EndFor
		\State \Return 
		\Else  ~\% Recursion step:
		\For{$i \in \{1,\ldots,\mathsf{d} \}$}
		\State $f_i \leftarrow i \cdot \frac{f_h-f_l}{\mathsf{d}} + f_l $
		\EndFor
		\For{$i \in \{1,\ldots,\mathsf{d} \}$}
		\State $ [r_i , r'_i] \leftarrow  \mathsf{TotalER}(\mathcal{CH} \cup Ch,f_i)$
		\State Calculate the reward $R_i = r_i + r'_i $
		\If{$R_i\geq R_{max}$}
		\State $R_{max} \leftarrow R_i$
		\State $f_{max} \leftarrow f_i$
		\EndIf
		\EndFor
		\For{$i \in \{1,\ldots,\mathsf{d} \}$}
		\State Calculate the possible maximum reward $\widetilde{R}_i= r_i \cdot \frac{ f_{i+1} }{f_i}+r'_{i+1}$
		\If{$\widetilde{R}_i>R_{max}$}
		\State $f_l \leftarrow f_i , f_h \leftarrow f_{i+1}$
		\State \Return $\f{cf}(\mathbf{LN},\mathcal{CH} \cup Ch,f_l,f_h)$
		\Else
		\State \% Do nothing - Discard this interval 
		\EndIf
		\EndFor
		\EndIf
		\EndFunction
		\State
		\Function{$\mathsf{TotalER}$}{$\mathcal{CH} \cup Ch,f$} 
		\State $r \leftarrow \mathsf{ER} (Ch,\fee{Ch}=f)$
		\State $r' \leftarrow \sum_{ \substack{ \forall Ch_j \in \mathcal{CH} } } \mathsf{ER} (Ch_j,\fee{Ch_j})$
		\State \Return $[r ,r']$
		\EndFunction
	\end{algorithmic}
\end{algorithm}

\end{document}